%% file: main.tex
\documentclass[11pt]{article}

\usepackage{tikz}
\usepackage{latexsym}
\usepackage{times}
\usepackage{amsthm}
\usepackage{amsfonts}
\usepackage{amssymb}
\usepackage{paralist}
\usepackage{amsmath}
\usepackage[letterpaper,margin=2.3cm]{geometry}
\usepackage{appendix}
\usepackage{graphicx}
\usepackage[boxed,figure,vlined]{algorithm2e}
\usepackage{enumerate}
\usepackage{tabularx}
\usepackage[compact]{titlesec}
\usepackage[capitalize]{cleveref}
\usepackage{subcaption}
\usepackage{color}
\usepackage[boxed,figure,vlined]{algorithm2e}
\usepackage[]{subfig}

\newcommand{\set}[1]{\left\{#1\right\}}
\newcommand{\floor}[1]{\left\lfloor #1 \right\rfloor}
\newcommand{\ceil}[1]{\left\lceil #1 \right\rceil}
\newcommand{\eps}{\varepsilon}

\DeclareMathOperator{\E}{\mathbb{E}}
\newcommand{\hide}[1]{}

\renewcommand{\include}{\input}

\newcommand{\todo}[1]{{\color{red}\bf [TODO: #1]}}
\renewcommand{\paragraph}[1]{\smallskip\noindent\textbf{#1: }}

\newcommand{\toblivious}{\tau}
\newcommand{\tphase}{\psi}

\newcommand{\newinformed}{\kappa}

\newtheorem{theorem}{Theorem}[section]
\newtheorem{lemma}[theorem]{Lemma}
\newtheorem{claim}[theorem]{Claim}

\newtheorem{definition}{Definition}[section]
\newtheorem{remark}{Remark}[section]

\newtheorem*{theoremRinner}{Theorem~\ref{\repthmref}}
\newenvironment{theoremR}[1]
  {\def\repthmref{#1}\theoremRinner (restated)}{\endtheoremRinner}
	
\title{The Cost of Global Broadcast in Dynamic Radio Networks\footnote{Research supported by ERC Grant No. 336495 (ACDC)}}
\author{Mohamad Ahmadi, Abdolhamid Ghodselahi, Fabian Kuhn, Anisur
  Rahaman Molla\\
Department of Computer Science\\
University of Freiburg\\
79110 Freiburg, Germany\\
{\small \texttt{\{mahmadi, hghods, kuhn, armolla\}@cs.uni-freiburg.de}}}

\date{}

\begin{document}

\maketitle              

\thispagestyle{empty}

\begin{abstract}
  We study the single-message broadcast problem in dynamic radio
  networks. We show that the time complexity of the problem depends on
  the amount of stability and connectivity of the dynamic network
  topology and on the adaptiveness of the adversary providing the
  dynamic topology. More formally, we model communication using the
  standard graph-based radio network model. To model the dynamic
  network, we use a generalization of the synchronous dynamic
  graph model introduced in [Kuhn et al., STOC 2010]. For integer
  parameters $T\geq 1$ and $k\geq 1$, we call a dynamic graph
  $T$-interval $k$-connected if for every interval of $T$ consecutive
  rounds, there exists a $k$-vertex-connected stable
  subgraph. Further, for an integer parameter $\tau\geq 0$, we say that
  the adversary providing the dynamic network is $\tau$-oblivious if for
  constructing the graph of some round $t$, the adversary has access
  to all the randomness (and states) of the algorithm up to round
  $t-\tau$.

  As our main result, we show that for any $T\geq 1$, any $k\geq 1$,
  and any $\tau\geq 1$, for a $\tau$-oblivious adversary, there is a
  distributed algorithm to broadcast a single message in time
  $O\big(\big(1+\frac{n}{k\cdot\min\set{\tau,T}}\big)\cdot n\log^3
  n\big)$. We further show that even for large interval
  $k$-connectivity, efficient broadcast is not possible for the usual
  adaptive adversaries. For a $1$-oblivious adversary, we show that
  even for any $T\leq (n/k)^{1-\eps}$ (for any constant $\eps>0$) and for
  any $k\geq 1$, global broadcast in $T$-interval $k$-connected
  networks requires at least $\Omega(n^2/(k^2\log n))$ time. Further,
  for a $0$-oblivious adversary, broadcast cannot be solved in
  $T$-interval $k$-connected networks as long as $T<n-k$.
\end{abstract}
\vspace{.5cm}

\input{intro}

\input{model}

\input{algorithm}

\input{lowerbound}

\input{impossibility}



\bibliographystyle{abbrv}
\bibliography{references}

\end{document}

%% file: intro.tex
\section{Introduction}
\label{sec:intro}

By now, a rich theory on algorithms for large-scale wireless networks
exists and we have a rather precise understanding of the complexity of
many basic computation and communication tasks for a variety of
wireless network models. While many wireless communication models and
modeling assumptions have been studied, to a large part, the
considered models all share one basic property. Most of the existing
work is based on static networks and on communication models where
wireless signal reception is modeled in a completely deterministic
way. For example, in the classic radio network model, a wireless
network is modeled as a graph and a node in the graph can receive a
message transmitted by some neighbor if and only if no other neighbor
transmits at the same time, e.g., \cite{bgi1,moscibroda05}. In the
SINR (or physical) model, nodes have fixed coordinates in some
geometric space and a transmitted signal can be successfully received
if and only if the signal-to-noise-and-interference ratio at the
receiver is above a certain fixed threshold, e.g.,
\cite{gupta:2000,moscibroda06}.

The situation in actual networks however is quite different and
wireless signal reception might behave in a rather unpredictable
way. There can be multiple sources for interference which cannot be
controlled by a distributed algorithm and signal propagation depends
on various properties of the environment. As a result, we often obtain
wireless communication links with unreliable behavior
\cite{kim06,newport:2007,ramachandran07,srinivasan08,yarvis02}. In
addition, wireless devices might be mobile leading to a potentially
completely dynamic network topology.

As a consequence, in recent years, researchers in the wireless
algorithms community have also started to consider radio network
models which exhibit nondeterministic behavior and sometimes general
dynamic topologies, e.g.,
\cite{structuringDG,clementi09,clementi12,fernandezanta12,dualgraph,obliviousDG}.
In the present paper, we continue this line of research and study the
global broadcast problem in dynamic radio networks for a range of
modeling assumptions. Note that in ordinary, static radio networks,
albeit appearingly simple, global broadcast is one of the best studied
problems in the area, (see, e.g.,
\cite{bgi2,bgi1,chlamtac:1985,jurdzinski14,newport-hitting} and many
others). We model a dynamic network by applying the dynamic network
model introduced in \cite{kuhn_stoc10}. Time is divided into
synchronous rounds and a wireless network is modeled as a dynamic
graph with a fixed set of $n$ nodes and a set of edges which can
change from round to round. For two parameters $T\geq 1$ and
$k\geq 1$, a dynamic graph is called \emph{$T$-interval $k$-connected}
if for any interval of $T$ consecutive rounds, the set of edges which
are present throughout these $T$ rounds induces a graph with vertex
connectivity at least $k$ (in \cite{kuhn_stoc10}, the model was only
introduced for $k=1$). We refer to \cite{kuhn:2011:survey} for a more
thorough discussion of the model of \cite{kuhn_stoc10} and of several
earlier related dynamic network models (e.g.,
\cite{avin08,baumann09,clementi09,odell05}).

Communication is modeled by using the standard
radio network model. In each round, each node can either transmit a message or
listen. A listening node successfully receives a message transmitted
by a neighbor in the current graph if and only if no other neighbor
transmits in the same round. We assume that nodes cannot detect
collisions, i.e., whether $0$ or more than $1$ neighbors transmit is
indistinguishable for a listening node. Note that the described
dynamic network model does not only allow to model topology changes
due to arbitrary node mobility. It also allows to model unreliable
links where the presence/availability can change for various reasons.

We assume that the dynamic graph is provided by a worst-case
adversary. As we study randomized distributed protocols, we need to
specify to what extent the adversary can adapt to the random decisions
of the nodes when determining the sequence of network topologies. For
the adaptiveness of the adversary, we use a more fine-grained
classification than what is usually done. For an integer parameter
$\toblivious\geq 0$, we say that the adversary is $\toblivious$-oblivious if for
determining the graph in round $r$, the adversary knows the randomness
of all nodes of all the rounds up to round $r-\toblivious$. Typically, only
the extreme cases are studied. An adversary which does not have access
to the random decisions of the algorithm (i.e., $\toblivious=\infty$) is
called an \emph{oblivious} adversary, whereas an adversary which has
access to the randomness of the algorithm is called an \emph{adaptive}
adversary. If the adversary even has access to the randomness of the
current round ($\toblivious=0$), it is called \emph{strongly adaptive},
otherwise ($\toblivious=1$), it is called \emph{weakly adaptive}. For more
precise formal definitions of the modeling assumptions, we refer to
\cref{sec:model}.

In our paper, we consider the problem of broadcasting a single message
from a source node to all the nodes of a dynamic network. The most
relevant previous work in the context of the present work appeared in
\cite{clementi09,obliviousDG,dualgraph}. In \cite{clementi09}, it is
shown that in $1$-interval $1$-connected networks (i.e., the graph is
connected in every round)\footnote{In \cite{clementi09}, the
  connectivity condition on the dynamic network is phrased differently
  and slightly more general.}, the complexity of global broadcast for
a $1$-oblivious adversary is $\Theta(n^2/\log n)$. In \cite{dualgraph}
and \cite{obliviousDG}, $\infty$-interval $1$-connected graphs are
considered (i.e., there is a stable connected subgraph which is
present throughout the whole execution). In \cite{dualgraph}, it is
shown that even for a $0$-oblivious adversary, it is possible to solve
broadcast in $O(n\log^2 n)$ rounds and it is shown that $\Omega(n)$
rounds are necessary even if the stable connected subgraph has
diameter $2$. In \cite{obliviousDG}, it is shown that when only
assuming an $\infty$-oblivious adversary, the running time can be
improved to $O((D+\log n)\log n)$, where $D$ is the diameter of the
stable connected subgraph. Note that in this case, the algorithm in
\cite{obliviousDG} achieves essentially the same time complexity as is
possible in static graphs of diameter $D$
\cite{bgi1,kushilevitz1998omega,newport-hitting}. In \cite{obliviousDG}, it is also
shown that for a $1$-oblivious adversary, $\Omega(n/\log n)$ rounds
are necessary even for $D=2$.

\subsection{Contributions}
\label{sec:contributions}

In the following, we state the results of the paper. For formal
details regarding problem statement and modeling not specified in the
introduction, we refer to Section \ref{sec:model}. Our main result is
a randomized broadcast algorithm for the described dynamic radio
network model. The algorithm (and also partly its analysis) is based
on a combination of the techniques used in \cite{clementi09} and
\cite{dualgraph}. We prove the following main theorem.

\begin{theorem}\label{thm:upperbound}
  Let $T\geq 1$, $\toblivious\geq 1$, and $k\geq 1$ be positive integer
  parameters. Assume that the adversary is $\toblivious$-oblivious. Then, in
  a dynamic $T$-interval $k$-connected $n$-node radio network, with
  high probability, single message broadcast can be solved in time
  \[
  O\left(\left(1+\frac{n}{k\cdot \min\set{\toblivious,T}}\right)\cdot n\log^3 n\right).
  \]
\end{theorem}
\begin{remark}
  Note that for small and for large values of $\min\set{\toblivious,T}$, one can
  do slightly better. It is straightforward to generalize the broadcast
  algorithm of \cite{clementi09} to complete single message broadcast
  in time $O\left(n^2/k\log n\right)$ in $1$-interval $k$-connected
  radio networks against a $1$-oblivious adversary. Using the result
  from \cite{dualgraph}, we also know that for a sufficiently large
  constant $c$ and $T\geq cn\log^2n$, single-message broadcast can be
  solved in $O(n\log^2n)$ rounds even for $\toblivious=0$. Our upper bound
  therefore beats previous results for
  $\min\set{\toblivious,T}=\omega(\log^4 n)$ and $T=O(n\log^2 n)$.
\end{remark}

In addition to the upper bound of Theorem \ref{thm:upperbound}, we
also prove a lower bound which essentially shows that even for very
large values of $T$, some relaxation on the standard adaptive
adversaries is necessary in order to get an upper bound which improves
with $T$. For $\toblivious=1$, we show that at least for small $k$, the
generalized upper bound of \cite{clementi09} is essentially
optimal. The lower bound can be seen as a generalization of the simple
$\Omega(n^2/\log n)$ lower bound for $k=1$ and $T=1$ proven in
\cite{clementi09}.

\begin{theorem}\label{thm:lowerbound}
  For every constant $\eps>0$ and every $T\leq (n/k)^{1-\eps}$, the
  expected time to solve single-message broadcast in $T$-interval
  $k$-connected radio networks against a $1$-oblivious adversary is at
  least
  \[
  \Omega\left(\frac{n^2}{k^2\log n}\right).
  \]
\end{theorem}

In addition, we show that unless the interval connectivity is very
large, single-message broadcast cannot be solved in the presence of a
strongly adaptive ($0$-oblivious) adversary.

\begin{theorem}\label{thm:impossibility}
  For any $k\geq 1$ and any $T<n-k$, it is not possible to solve
  single-message broadcast in $T$-interval $k$-connected radio
  networks against a $0$-oblivious adversary.
\end{theorem}

The discussion of the above result appears in
Section \ref{sec:impossibility}.  We note that the above theorem is
tight in the following sense. As soon as $T\geq n-k$, global broadcast
can be solved (with potentially exponential time complexity) and as
soon as $T = cn\log^2n$ for a sufficiently large constant $c$, we know
from \cite{dualgraph} that it can even be solved in time
$O(n\log^2 n)$. All results, as well as a comparison with previous
work are summarized in \Cref{tab:results}.

\begin{table}[t]
\centering
  \begin{tabular}{ | c | c | c | c | c |}
    \hline
    interval conn. & vertex conn. & adversary & complexity \\ \hline \hline
    $T=1$ & $k=1$ & $\toblivious=1$ & $\Theta\left(n^2/\log n\right)$ \cite{clementi09} \\
    \hline

    $T=\infty$ & $k=1$ & $\toblivious=0$ & $O\left(n\log^2n\right)$ / $\Omega(n)$, D=2\cite{dualgraph} \\
    \hline

    $T=\infty$ & $k=1$ & $\toblivious=\infty$ & $O\left((D+\log n) \cdot \log n\right)$\cite{obliviousDG} \\
    \hline

    $T=\infty$ & $k=1$ & $\toblivious=1$ & $\Omega(n/\log n)$, D=2 \cite{obliviousDG}\\
    \hline

    \boldmath$T\geq1$ & \boldmath$k\geq1$ &
    \boldmath$\toblivious\geq1$ &
    {\boldmath$O\left(\left(1+\frac{n}{k\cdot
            \min\set{\toblivious,T}}\right)\cdot n\log^3 n\right)$} \\
    \hline
    \boldmath$T \leq (n/k)^{1-\eps}$ & \boldmath$k \geq 1$ & 
    \boldmath$\toblivious=1$ & {\boldmath$\Omega\left(n^2/(k^2\log n)\right)$} \\
    \hline
    \boldmath$T < n-k$ & \boldmath$k \geq 1$ &
    \boldmath$\toblivious=0$ & \textit{\textbf{impossible}}\\
    \hline

  \end{tabular}
  \caption{An overview over the existing bounds on global broadcast in
    the dynamic radio network model. The results marked in bold are
    the results of the present paper. For the $T=\infty$ results, $D$
    refers to the diameter of the stable subgraph.}
  \label{tab:results}
\end{table}

\begin{remark} In \cite{kuhn_stoc10}, interval connectivity was
introduced (in particular) to study the problem of broadcasting
multiple messages in a dynamic network in a standard message passing
model. It is shown that interval connectivity $T$ allows to
essentially speed up multi-message broadcast by a factor of $T$. We
find it interesting that when considering a radio network model,
interval connectivity seems to provide a similar speed-up, even for
broadcasting a single message. Something similar also holds for graphs
with large vertex connectivity. In \cite{CGK-PODC14,CGK}, it is shown
that even on static graphs, vertex connectivity $k$ allows to speed up
multi message broadcast by essentially a factor $k$. Here, we show
that a similar speed up can be obtained in radio networks even for
broadcasting a single message.
\end{remark}


%% file: model.tex
\section{Model and Problem Definition}
\label{sec:model}

\paragraph{Dynamic Network}
As described in Section \ref{sec:intro}, we adapt the synchronous
dynamic network model of \cite{kuhn_stoc10} to model dynamic
networks.\footnote{Similar dynamic network models have also been used
  before \cite{kuhn_stoc10}, for example in
  \cite{avin08,clementi09,odell05}.  For additional references and a
  thorough discussion, we refer to \cite{kuhn:2011:survey}.} Time is
divided into rounds such that for all $r\geq 1$, round $r$ starts at
time $r-1$ and ends at time $r$. A dynamic network is given by a
sequence of undirected graphs $\langle G_1, G_2, \dots\rangle$, where
$G_r=(V,E_r)$ is a static graph representing the network topology in
round $r$. The node set $V$ is a set of $n$ nodes corresponding to the
wireless devices in the network and the edge set $E_r$ is the set of
active communication links in round $r$. A dynamic graph
$\langle G_1,G_2,\dots\rangle$ is called \textit{$T$-interval
  $k$-connected} for integer parameters $T\geq 1$ and $k\geq 1$ if and
only if for all $r\geq 1$, the graph
\[
\bar{G}_{r,T} = (V,\bar{E}_{r,T}), \qquad 
\bar{E}_{r,T} := \bigcap_{r'=r}^{r+T-1} E_{r'}
\]
is a graph with vertex connectivity at least $k$.
	
\paragraph{Communication Model}
An $n$-node distributed algorithm $\mathcal{A}$ is defined by $n$
randomized processes which are assigned to the nodes of the dynamic
graph by an adversary.  For simplicity we use the term \textit{node
  $u$} to also refer to the process which is assigned to node $u$. In
each round, each node decides either to transmit a message or to
listen to the wireless channel. The behavior of the wireless channel
is modeled by using the standard radio network model first used in
\cite{bgi1,chlamtac:1985}. When node $u$ decides to transmit in round
$r$, its message reaches all of its neighbours in $G_r$.  A node $v$
which listens in round $r$ receives a message transmitted by a
neighbor $u$ if and only if $u$ is the only neighbor of $v$ in $G_r$
which is transmitting in round $r$. If no message reaches $v$ (no
neighbor is transmitting), $v$ receives silence, indicated by
$\bot$. If two or more messages reach $v$, $v$ also receives $\bot$,
i.e., $v$ cannot distinguish $2$ or more transmitting neighbors from
silence.

	
\paragraph{Adversary}
We assume that the network changes under the control of an
adversary. For any round $r$ the adversary has to determine $G_r$
based on the knowledge it has. For an integer $\tau \geq 0$, we call
an adversary \textit{$\tau$-oblivious} if for any $r\geq 1$, the
adversary constructs $G_r$ based on the knowledge of: (1) the algorithm
description, (2) the network topologies of rounds $1,\dots,r-1$, and
(3) the nodes' random choices of the first $r-\tau$ rounds.

\paragraph{Global Broadcast}
A distributed algorithm solving the global broadcast problem needs to
disseminate a single message $\mathcal{M}$ from a distinguished source
node to all the processes in the network. We assume that in a
distributed broadcast algorithm, non-source nodes are activated (and
can start to actively transmit) when they first receive the broadcast
message $\mathcal{M}$. Nodes that do not yet know $\mathcal{M}$ remain
silent.

\paragraph{Mathematical Notation}
For two integers $a\leq b$, $[a,b]$ denotes the set of all integers
between $a$ and $b$ (including $a$ and $b$). Further, for an integer
$a\geq 1$, we use $[a]$ as a short form to denote $[a]:=[1,a]$. We say
that a probability event happens with high probability (w.h.p.) if it
happens with probability at least $1-1/n^c$, where $n$ is the number
of nodes and $c>0$ is a constant which can be chosen arbitrarily large
by adjusting other constants.


%% file: algorithm.tex
\section{Upper Bound}
\label{sec:algorithm}

\subsection{Randomized Broadcasting Algorithm}
We now describe our randomized algorithm which solves broadcast in a
$T$-interval $k$-connected radio network against a
$\toblivious$-oblivious adversary. As stated in
\Cref{sec:contributions}, the algorithm has a time complexity of
$O\left(\left(1+n/(k\tphase)\right)\cdot n\log^3n\right)$ with high
probability where $\tphase:=\min\set{\toblivious,T,n/2k}$. In light of the
comment following Theorem \ref{thm:upperbound} in
\Cref{sec:contributions}, throughout \Cref{sec:algorithm}, we assume
that $\tphase=\Omega(\log^3 n)$ as otherwise, one can achieve a
stronger upper bound by just using an adapted version of
\cite{clementi09}. 

In the first round, the source node transmits the message to its
neighbors. Because we assume that each graph is $k$-vertex connected,
after one round, at least $k+1$ nodes know the message. From there on,
our randomized algorithm works in \textit{phases}. To simplify
notation, in the following, we ignore the first round and assume that
at time $0$, the algorithm starts with at least $k+1$ nodes which know
the broadcast message $\mathcal{M}$. The phases of the algorithm are
defined as follows.

\begin{definition}[Phase]\label{de:phase}
  The $j^{th}$ time interval of $\tphase$ consecutive rounds is called
  phase $j$, where $j$ is a positive integer. Hence, phase $j$ starts
  at time $(j-1)\tphase$ and ends at time $j\tphase$ and it consists of
  rounds $(j-1)\tphase+1,\dots,j\tphase$.
\end{definition} 

Let $t_v$ denote the round in which $\mathcal{M}$ is received by node
$v$ for the first time. In each round $t$ the set $V$ is partitioned
into following three subsets. The \textit{previously informed nodes}
$I(t)$ are the nodes that have received $\mathcal{M}$ in some phase
before the current phase.
Note that in the first phase, $I(t)$ consists of at least $k+1$
informed nodes. The nodes that have received $\mathcal{M}$ for the first
time in the current phase in some round before time $t$ are called
\textit{newly informed nodes}, and they are denoted by $N(t)$.
Finally, the set of \textit{uninformed nodes} at time $t$ is denoted
by $U(t):=V \setminus \set{I(t) \cup N(t)}$.

The algorithm can be seen as a combination of two existing protocols
which appeared in \cite{clementi09} and \cite{dualgraph}. The protocol
of \cite{clementi09} is a very basic one where all informed nodes
always try to transmit the message independently with the same uniform
probability. In the \emph{harmonic broadcast} protocol of
\cite{dualgraph}, informed nodes use harmonically decaying
probabilities to forward the message. In each phase of our algorithm,
in the first $\ceil{\tphase/2}$ rounds, a variant of the protocol of
\cite{clementi09} is applied and in the second $\floor{\tphase/2}$
rounds, the idea of the protocol of \cite{dualgraph} is applied. In
the following, the algorithm is described in detail.

\paragraph{First half of a phase} In the first $\ceil{\tphase/2}$
rounds of a phase, all informed nodes, i.e., all $ v \in \{I(t) \cup
N(t)\}$, transmit the message with probability $1/n$. 

\paragraph{Second half of a phase}
The nodes in $U(t)\cup I(t)$ remain silent throughout the second half
of a phase.  However, in each round $t$, any node $v \in N(t)$
transmits the message with probability $p_v(t)$, given by
		
\begin{equation}\label{eq:harmonicalg}
  \forall t > \floor{ \frac{t}{\tphase}} \cdot \tphase
  +\ceil{\frac{\tphase}{2}}\ :\ \forall v \in N(t)\, :\,
  p_v(t) := \frac{1}{1+ \floor{\frac{t-\hat{t}_v-1}{\mathcal{T}}}},
\end{equation}
where 
\[
\hat{t}_v := \begin{cases}
  \floor{ \frac{t}{\tphase}} \cdot \tphase
  +\ceil{\frac{\tphase}{2}}, & \text{ if } \floor{ t/\tphase} \cdot
  \tphase < t_v < \floor{t/\tphase} \cdot \tphase +\ceil{\tphase/2}\\
  t_v, & \text{ otherwise}
\end{cases}
\]
and $\mathcal{T}$ will be fixed in Lemma~\ref{le:isolated}. 

Thus, in the second half of a phase, only nodes participate which for
the first time receive $\mathcal{M}$ in the current phase. Each node
$v$ which gets newly informed in the phase executes the following
protocol. As soon as $v$ knows $\mathcal{M}$ and as soon as the second
half of the phase has started, $v$ starts transmitting $\mathcal{M}$
to its neighbors.  For the first $\mathcal{T}$ rounds, $v$ transmits
the $\mathcal{M}$ with probability 1, for the next $\mathcal{T}$
rounds $v$ transmits $\mathcal{M}$ with probability $1/2$, and the
probability for the next time intervals of $\mathcal{T}$ rounds
becomes $1/3$,$1/4$, etc.

\subsection{Analysis}

Recall that by the definitiuon of $\tphase$, we have $k\tphase\leq n/2$,
$\tphase\leq T$, and $\tphase \leq \toblivious$. The $T$-interval
$k$-connectivity of the dynamic network guarantees the existence of a
stable spanning subgraph with vertex connectivity of at least $k$
throughout the whole duration of every phase. We call this reliable
spanning subgraph the \textit{backbone} of the phase. Note that we may
have different backbones in different phases. Let $P(t)$ denote the
sum of transmitting probabilities of all the nodes in round $t$, i.e.,
$P(t):=\sum_{v\in V} p_v(t)$. For the analysis of our algorithm, we
say that round $t$ is \textit{busy} if $P(t) \geq 1$ and otherwise we
say that round $t$ is \textit{free}.
If the node $v$ is the only node transmitting in a round, we say that
node $v$ gets \textit{isolated} in that round.

For any phase $j$, let $\theta_0 :=j\tphase - \floor{\tphase/2}$,
i.e., $\theta_0$ is the time when the second half of the phase
starts. For $i>0$, we define $\theta_i>\theta_{i-1}$ to be the first
time such that in the time interval $[\theta_{i-1},\theta_i]$ (i.e.,
in rounds $\theta_{i-1}+1,\dots,\theta_i$) the number of busy rounds
equals the number of free rounds (see Figure \ref{fig:thetas}).

\begin{figure}[t]
\centering
\includegraphics[scale=0.9]{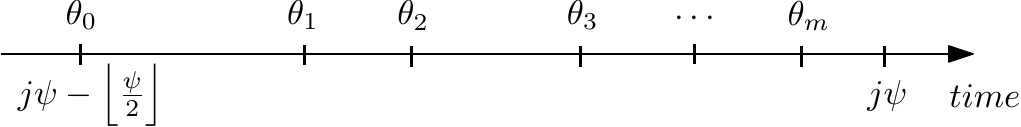}
\caption{Time intervals $[\theta_{i-1},\theta_i]$ with equal number of
  free and busy rounds where $i >1$.}
\label{fig:thetas}
\end{figure}

We further define $m\geq 0$ such that $\theta_m$ is the last such time
defined for a given phase. The case $m=0$ implies that throughout the
second half of the phase, the number of busy rounds is always larger
than the number of free rounds.

We use the following lemma adapted from Lemma 13 of \cite{dualgraph}.
\begin{lemma}\label{le:isolated}\cite{dualgraph}
  Consider a node $v$. Let $t > \hat{t}_v$ be such that at least half
  of the rounds $\hat{t}_v+1,\dots,t$ are free.  If
  $\mathcal{T} \geq \ceil{12\ln(n/\epsilon)}$ for some $\epsilon>0$,
  then with probability larger than $1-\epsilon/n$ there exists a
  round $t'\in \set{\hat{t}_v+1,\dots,t}$ such that $v$ is isolated in
  round $t'$.
\end{lemma}

\begin{lemma}\label{le:isolatedintheta}
  For all phases, in each time interval
  $\left[\theta_{i-1},\theta_i\right]$, where $i \in [m]$, if round
  $\theta_{i-1}+1$ is busy then any node $v$ with
  $\hat{t}_v \in \set{\theta_{i-1},\dots.\theta_i-1}$ gets isolated in some
  round $t'\in \set{\hat{t}_v+1,\dots,\theta_i}$ with high probability.
\end{lemma}
\begin{proof}
  Let $\bar{t}$ denote the first round that the number of free rounds
  equals the number of busy rounds starting from round
  $\hat{t}_v+1$. For the sake of contradiction, assume that
  $\bar{t} > \theta_i$, that is, the number of free rounds is less
  than the number of busy rounds in $\set{\hat{t}_v+1,\dots,\theta_i}$
  and we also know that the number of busy rounds is greater than the
  number of free rounds in $\set{\theta_{i-1}+1,\dots,\hat{t}_v}$
  (because of minimality of $\theta_i$ and the fact that round
  $\theta_{i-1}+1$ is busy). It follows that the number of busy rounds
  is greater than the number of free rounds in
  $\set{\theta_{i-1}+1,\dots,\theta_i}$ contradicting our assumption
  on the equality of free and busy rounds in
  $\set{\theta_{i-1}+1,\dots,\theta_i}$. Therefore,
  $\bar{t} \leq \theta_i$ and according to Lemma~\ref{le:isolated} the
  claim holds.
\end{proof}

As one can see in Figure \ref{fig:uninfcomp}, at the beginning of each
phase, the uninformed nodes in the backbone form one or several
connected subgraphs which we call the \textit{uninformed connected
  components}.  For each uninformed connected component there must
exist some edge in the backbone (within a phase) connecting an
informed node to a node in that component. Note that because the
adversary is $\toblivious$-oblivious and thus also oblivious to the
last $\tphase\leq \toblivious$ rounds, the adversary has to determine the
dynamic graph throughout a phase before the phase starts. The backbone
graph of a phase can therefore not change depending on the randomness
of the algorithm during the phase.

\begin{definition}[Available Components and Available Nodes]\label{de:availablenodes}
  At the end of the first $\ceil{\tphase/2}$ rounds of each phase, any
  uninformed connected component that includes at least one newly
  informed node is called an available component. All the nodes in an
  available component are called available nodes.
\end{definition}

\begin{figure}[t]
\centering
\includegraphics{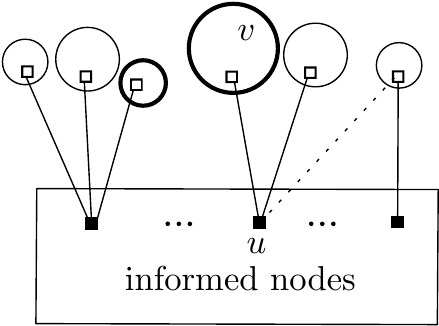}
\caption{Backbone of a phase. Available components are identified by thick circles.}
\label{fig:uninfcomp}
\end{figure}

\begin{lemma}\label{le:broadcastproceed}
  Consider an arbitrary phase and an arbitrary $i\geq 1$. If at the
  beginning of round
  $\theta_i$ of the phase there exists at least one uninformed
  available node, then w.h.p.\ at least one available node gets informed in
  round $\theta_i$.
\end{lemma}
\begin{proof} We will show that for every $i\geq 1$, w.h.p., if there
  is some node available $u$ with $\hat{t}_u = \theta_{i-1}$ and at
  the beginning of round $\theta_i$, there is at least one uninformed
  available node, then at least one available node $v$ gets informed
  in round $\theta_i$. The claim of the lemma then follows by
  induction on $i$. If there are no available nodes, there is nothing
  to prove. If there are available nodes, there is at least one node
  $u$ which gets newly informed in the first half of the phase and we
  therefore have $\hat{t}_u=\theta_0$. Using the above claim, it then
  w.h.p.\ follows that if there still is an uninformed available node
  at time $\theta_1-1$, some uninformed available node $u'$ gets
  informed in round $\theta_1$ and thus $\hat{t}_{u'} = \theta_1$. For
  $i>1$, the induction step now follows in the same way.  It therefore
  remains to show that w.h.p., if there is some node available $u$
  with $\hat{t}_u = \theta_{i-1}$ and at the beginning of round
  $\theta_i$, there is at least one uninformed available node, then at
  least one available node $v$ gets informed in round $\theta_i$.

  By Lemma \ref{le:isolatedintheta} we know that w.h.p., all the nodes
  $u$ with $\hat{t}_u\in \set{\theta_{i-1},\dots,\theta_{i}-1}$ get
  isolated in some round
  $t'\in \set{\theta_{i-1}+1,\dots,\theta_{i}}$. Hence, by induction
  on $j$, w.h.p., for all $j\leq i-1$ there is some node $u'$ with
  $\hat{t}_{u'}=\theta_j$ and therefore all nodes $u$ with
  $\hat{t}_u\in\set{\theta_0,\dots,\theta_i-1}$ get isolated in some
  round $t'\in \set{\theta_0+1,\dots,\theta_i}$. Consequently, w.h.p.,
  all newly informed nodes $N(\theta_i-1)$ at time $\theta_i-1$ get
  isolated in some round $t'\in \set{\theta_0+1,\dots,\theta_i}$. Let
  $v$ be an uninformed available node before round $\theta_i$ (i.e.,
  at time $\theta_i-1$). Because $v$ is available, at time
  $\theta_i-1$, there is a informed available neighbor $u$ in the
  backbone graph of the current phase. We clearly have
  $u\in N(\theta_i-1)$ and thus w.h.p., $u$ gets isolated in some
  round $t'\in \set{\theta_0+1,\dots,\theta_i}$. As soon as $u$ gets
  isolated, $v$ gets informed and we can therefore conclude that $u$
  gets isolated in round $\theta_i$ and thus $v$ gets informed in
  round $\theta_i$.
\end{proof}

\begin{lemma}\label{le:informednodesbounded} 
  Consider an arbitrary phase and assume that at the beginning of the
  second $\floor{\tphase/2}$ rounds of the phase there are $z$
  available nodes. Then, w.h.p., for some constant $c>0$, at the end
  of the phase we have at least $\min\set{z,c\tphase/ \ln^2 n}$ newly
  informed nodes.
\end{lemma}
\begin{proof}
  We show that the claim holds through two general cases which could
happen within any phase, Fix some phase $j$.
  
\paragraph{Case 1} The number of busy rounds is not less than the
number of free rounds in the second $\floor{\tphase/2}$ rounds of
phase $j$.
  
\begin{claim}\label{cl:morebusy}
  If the number of busy rounds is not less than the number of free
  rounds in the second $\floor{\tphase/2}$ rounds of phase
  $j$, then the number of newly informed nodes at the end of the phase
  is at least $c\tphase / \ln^2 n$.
\end{claim}
\begin{proof}
  The sum of transmitting probabilities of all nodes in all rounds in
  the second $\floor{\tphase/2}$ rounds of phase $j$ is not smaller
  than $\floor{\tphase / 4}$, i.e.
  \begin{equation}\label{eq:probabilitysumlowbound}
    \sum\limits_{t=j\tphase - \floor{\frac{\tphase}{2}}+1}^{j\tphase} P(t) \geq   \floor{\frac{\tphase}{4}}.
  \end{equation} 
  Let $\newinformed$ denote the number of newly informed nodes at the
  end of the phase. We have
  \begin{equation}\label{eq:probabilitysumupbound}
    \sum\limits_{t=j\tphase - \floor{\frac{\tphase}{2}}+1}^{j\tphase} P(t) \stackrel{\eqref{eq:harmonicalg}}{\leq} \newinformed \mathcal{T} \cdot \sum\limits_{i=1}^{\ceil{\frac
        {\tphase}{2\mathcal{T}}}} \frac{1}{i} < \newinformed \mathcal{T} \cdot \left(\ln\ceil{\frac{\tphase}{2\mathcal{T}}}+1\right).
  \end{equation} 
  Then (\ref{eq:probabilitysumlowbound}) and
  (\ref{eq:probabilitysumupbound}) lead to
  \[
  \newinformed > \frac{\floor{\frac{\tphase}{4}}}{\mathcal{T} \cdot
    \left(\ln\ceil{\frac{\tphase}{2\mathcal{T}}}+1\right)}.
  \]
  There exists some constant $c > 0$ such that the right-hand side of
  the above inequality is not smaller than $c\tphase/\ln^2 n$ and thus the
  claim follows.
\end{proof}

\paragraph{Case 2} The number of busy rounds is less than the number
of free rounds in the second $\floor{\tphase/2}$ rounds of phase $j$.
  
In this case, round $\theta_m+1$ for $m \geq 0$ cannot be
busy. Otherwise the number of busy rounds is larger than the number of
free rounds in $[\theta_m+1,j\tphase]$ because $\theta_m$ is supposed
to be the last round of such intervals of second $\floor{\tphase/2}$
rounds as it is shown in Figure \ref{fig:thetas}. This implies that
the number of busy rounds is larger than the number of free rounds in
the second $\floor{\tphase/2}$ rounds of the phase which contradicts
the assumption of Case 2. Therefore, round $\theta_m+1$ for $m \geq 0$
must be free. For the sake of contradiction suppose that none of $z$
available nodes gets informed. Therefore, we can conclude that there
exists some available uninformed node in round $\theta_m$ and by Lemma
\ref{le:broadcastproceed} one node gets informed in this round which
leads to having the round $\theta_m+1$ busy that contradicts our
assumption on its freeness. As a result, all the $z$ available nodes
get informed.

Case 1 and Case 2 together conclude the claim.
\end{proof}


Using the established technical lemmas, we can now proof our upper
bound theorem.

\begin{theoremR}{thm:upperbound}
  Let $T\geq 1$, $\toblivious\geq 1$, and $k\geq 1$ be positive integer
  parameters. Assume that the adversary is $\toblivious$-oblivious. Then, in
  a dynamic $T$-interval $k$-connected $n$-node radio network, with
  high probability, single message broadcast can be solved in time
  \[
  O\left(\left(1+\frac{n}{k\cdot \min\set{\toblivious,T}}\right)\cdot n\log^3 n\right).
  \]
\end{theoremR}
\begin{proof}
  Consider some phase $j$ and let $\mathcal{B}_j$ be the backbone of
  phase $j$, i.e., $\mathcal{B}_j$ is the stable $k$-connected
  subgraph of phase $j$.  Consider the subgraph $\mathcal{B}_j[U_j]$
  of $\mathcal{B}_j$ induced by the uninformed nodes $U_j$. This
  induced subgraph might consist of several connected
  components. However, each of the components is connected to at least
  $k$ nodes in $I_j$ as it is shown in Figure \ref{fig:uninfcomp} (recall that we can assume that $|I(t)|\geq
  k+1$). Note that if one of these at least $k$ nodes gets isolated in
  the first half of the phase, all nodes in the component become
  available for the second half of the phase.
  
  In the first half of phase $j$, in any round $t$, each node in
  $I(t)$ transmits the message with probability $1/n$. Therefore,
  for every node $u\in I(t)$, the probability that $u$ gets isolated
  in round $t+1$ (in the first half of a phase) is at least
  
  \begin{equation}
    \label{eq:nodeisolation}
    \Pr(u\text{ gets isolated in round }t+1) \geq
    \frac{1}{n}\left(1-\frac{1}{n}\right)^{n-1} > \frac{1}{en}.
  \end{equation}
  
  In following, we analyze the progress in the first half of some
  phase $j$. Consider an uninformed node $v\in U_j$ (at the beginning
  of phase $j$). Let $\mathcal{A}_v$ be the event that $v$ becomes
  available in the first half of phase $j$. Event $\mathcal{A}_v$
  definitely occurs if one of the at least $k$ initially informed
  neighbors of $v$'s component in $\mathcal{B}_j[U_j]$ gets isolated
  in one of the at least $\tphase/2$ rounds of the first half of the
  phase. The probability for this is
  \[
  \Pr(\mathcal{A}_v) \geq 1 - \left(1-\frac{k}{en}\right)^{\tphase/2}
  > 1 - e^{-k\tphase/2en} \geq \frac{k\tphase}{4en}.
  \]
  The last inequality follows from the fact that for all $0\leq x\leq 1$,
  $e^{-x}\leq 1-x/2$. Let $X$ be the number of nodes in $U_j$ that get
  available in phase $j$. For convenience, we define $\lambda:=
  |U_j|/n \leq 1$. We then have
  \begin{equation}
    \label{eq:availableexpect}
    \E[X] = \sum_{v\in U_j} \Pr(\mathcal{A}_v) \geq |U_j|\cdot
    \frac{k\tphase}{4en} = \frac{\lambda k \tphase}{4e}.
  \end{equation}
  We define $F:= \min\big\{\frac{\tphase}{16e},\frac{c\tphase}{\ln^2
      n}\big\}$, where $c>0$ is the constant that is used in Lemma
  \ref{le:informednodesbounded}.  Note that by Lemma
  \ref{le:informednodesbounded}, in phase $j$, w.h.p., at least
  $\min\{X,F\}$ uninformed nodes become informed.

  We define a phase to be successful if $X \geq \lambda F$. Let
  $\mathcal{S}$ be the event that phase $j$ is successful and let
  $\bar{\mathcal{S}}$ be the complementary event. We can upper bound
  the expected value of $X$ as follows:
  \[
  \E[X] < \Pr(\bar{\mathcal{S}})\cdot \lambda F +
  \big(1- \Pr(\bar{\mathcal{S}})\big)\cdot \lambda n.
  \]
  Combining with the upper bound in \eqref{eq:availableexpect},
  we obtain (recall that we assume that $k\tphase \leq n/2$).
  \begin{equation}
    \label{eq:nosuccess}
    \Pr(\bar{\mathcal{S}}) < 
    \frac{n - \frac{k\tphase}{4e}}{n - F} \leq
    \frac{n - \frac{k\tphase}{4e}}{n - \frac{\tphase}{16e}} 
    \leq
    \left(1-\frac{k\tphase}{4en}\right)\left(1+\frac{\tphase}{8en}\right)
    \stackrel{(k\geq 1)}{\leq}
    \left(1-\frac{k\tphase}{8en}\right).
  \end{equation}
  
  By Lemma \ref{le:informednodesbounded}, in a successful phase,
  w.h.p., at least $\lambda F$ new nodes get informed. Hence, in a
  successful phase, w.h.p., we get rid of at least an $(F/n)$-fraction
  of the remaining uninformed nodes. In order to inform all nodes,
  w.h.p., we therefore need at most
  $O(n\log(n)/F)=O(n\log^3(n)/\tphase)$ successful phases. Using
  \eqref{eq:nosuccess} and a standard Chernoff argument, we can thus
  w.h.p. upper bound the total number of phases by
  $O\big(n^2\log^3(n)/(k\tphase^2)\big)$. As each phase takes
  $\tphase$ round, this concludes the proof.
\end{proof}


%% file: lowerbound.tex
\section{Lower Bound}
\label{sec:lower}

In this section we prove a lower bound for global broadcast in
$T$-interval $k$-connected radio networks against a $1$-oblivious
adversary.  Furthermore, we show impossibility of solving the same
problem against a strongly adaptive adversary ($0$-oblivious
adversary).

Our lower bound is based on a general technique for proving lower
bounds for communication problems in radio networks, introduced by
Newport in \cite{newport-hitting}.  Using this technique, one first
defines a combinatorial game for which a lower bound can be proved
directly. It is then shown how to reduce the game to the problem in
order to leverage the game's lower bound to obtain the desired lower
bound for the problem.

To prove \Cref{thm:lowerbound} using this technique, we first
introduce an abstract hitting game, called the
\textit{$(\beta,\ell,\varphi)$-periodic hitting game} and directly
prove a lower bound for winning this game. We note that this game is
more involved than the games used in previous work,e.g.,
\cite{obliviousDG,newport-hitting}. Based on a lower bound for the
hitting game, for a given instance of the game we instantiate an
$n$-node \emph{target network}. By instantiation of an $n$-node network, we
mean assigning $n$ processes with unique IDs to the nodes of the
network.  For the instantiation, one needs to also have information
which is not available to the player in the game. However, we show
that by playing the game, the player can still simulate the execution
of a given broadcast algorithm on the corresponding target network to
the given instance of the game.  We show that this simulation of a
broadcast algorithm allows to win the hitting game and the lower bound
on the hitting game and the simulation together imply a lower bound
for solving global broadcast.

\paragraph{\boldmath$(\beta,\ell,\varphi)$-periodic hitting game}
The game is defined for three integers $\beta ,\ell,\varphi >0$ and
proceeds in rounds.  Time is divided into phases of $\varphi$ rounds,
where the $j^{th}$ phase of the game is called phase $\pi_j$.  That
is, phase $\pi_1$ consists of rounds $1$ to $\varphi$, phase $\pi_2$
consists of rounds $\varphi +1$ to $2\varphi$, etc.  The first round
of any phase $\pi_j$ is called $t_j$.  The player of the game is
represented by a probabilistic automaton $\mathcal{P}$ and plays the
game against a referee.  Two sets are defined for this game, a
\textit{selection set} $S :=[\beta]$ which is fixed during the game,
and a \textit{target set} which might change from round to round. The
target set of round $t$ is denoted by $X(t)$.  In each round $t$,
$\mathcal{P}$ chooses one element from $S$ and outputs it as the guess
$\gamma(t)$ for round $t$. Round $t$ is called a \textit{successful
  round}, if and only if $\gamma (t) \in X(t)$.
		
At the beginning of each phase $\pi_j$ ($j\geq 1$), the referee
determines a set $Y_j$ consisting of $\ell$ elements chosen uniformly
at random from $S$.  We define the target set as follows. For
convenience, assume that $Y_0$, $Y_{-1}$ and $X(0)$ are empty sets.
\begin{equation*}
  \forall j, \forall t\in \left[t_j,t_j+\varphi-1\right]\,:\,
  X(t) :=
  \begin{cases}
    Y_j \cup \big[X(t-1) \setminus (\{\gamma(t-1)\}\cup
    (Y_{j-2}\setminus Y_{j-1}))\big] &
    \text{if }{t=t_j},\\
    X(t-1)\setminus \{\gamma(t-1)\} & \text{if }{t\neq t_j}.
  \end{cases}
\end{equation*}  

\noindent
That is, at the beginning of each phase $\pi_j$, the referee chooses
$\ell$ elements $Y_j$ from $S$ uniformly at random and adds them to
the target set.  Two phases ($2\varphi$ rounds) later, each of these
$\ell$ elements which still remains in the target set (and which is
not in $Y_{j+1}$) is removed from the target set by the referee.
Moreover, after each successful round, the referee removes the correct
guess from the target set.  Player $\mathcal{P}$ wins the game in $r$
rounds if and only if either round $r$ is the $\beta^{th}$ successful round
for $\mathcal{P}$, or before round $r+1$ (in phase $j$),
$X(r+1)\cap Y_{j-1}=\emptyset$ or $X(r+1)\cap Y_j=\emptyset$. The
second condition will be used to ensure sufficiently large interval
connectivity of the target network as long as the game is not won. The
only information that the player receives at the end of each round is
whether the round was successful or not. The player is also notified
if it wins the game.
		
Intuitively, as long as the target set changes sufficiently often, it
should always appear essentially random to the player. Therefore, the
best strategy for hitting the target set is to always choose an almost
uniformly random guess, leading to roughly $\beta/\ell$ rounds to get
a single successful round.  The following lemma states this intuition
formally. 
\begin{lemma}\label{le:gamelb}
  For any $\varphi \leq \beta/3$ and for $\ell\geq \ell_0$ for a
  sufficiently large constant $\ell_0>0$, the expected number of
  rounds for a player to win the ($\beta,\ell,\varphi$)-periodic
  hitting game is at least $\Omega(\beta^2/\ell)$.
\end{lemma}

\begin{proof}
From the game's definition, we can conclude that for winning the game, it is necessary for the player to either experiences $\beta$ successful rounds or there exist two consecutive phases ($2\varphi$ rounds) consisting at least $\ell$ successful rounds. We first show that for sufficiently large constant $\ell$, the latter case does not happen with high probability. Then, we show that the expected number of rounds for the player to experience $\beta$ successful rounds is at least $\Omega (\beta^2/\ell)$.

Fix some two consecutive phases ($2\varphi$ rounds). We assume that $\varphi < \beta ^{1-\eps}$ for some constant $\eps > 0$. Let random variable $Z$ be the number of successful rounds in these $2\varphi$ rounds. Then $Z$ is dominated by a binomial random variable with parameters $2\varphi$ and $2\ell/(\beta -2\varphi +1)$. Then, 
\[
\Pr\left(Z\geq \ell \right) \leq 
	\binom{2\varphi}{\ell}\cdot \left( \frac{2\ell}{\beta - 2\varphi +1}\right) ^\ell \leq		
		\left( \frac{2e\varphi}{\ell} \right) ^\ell \cdot \left( \frac{2\ell}{\beta - 2\varphi +1}\right) ^\ell \leq
			\left( \frac{8e\varphi}{\beta} \right)^\ell \leq
				\left( \frac{8e}{\beta ^\eps} \right)^\ell .
\]
It is, therefore, sufficient to choose $\ell = O(1/\eps) = O(1)$ so that by applying union bound over all $O(\beta^2)$ such periods of $2\varphi $ rounds, we guarantee that with high probability there do not exist two consecutive phases consisting at least $\ell$ successful rounds throughout the whole execution. 

Now let us calculate the expected number rounds for the player to experience $\beta$ successful rounds. In any round $t$ we have at most $2\ell$ elements in the target set
and hence $|X(t)|\leq 2\ell$. Let $\Gamma_j$ be the set of all
proposals in phase $\pi_j$:
\begin{equation*}
  \Gamma_j = \bigcup \limits_{i = t_j}^{t_j+\varphi -1} \{\gamma(i)\}
\end{equation*}
and $\Gamma_j(t)$ the set of all proposals in phase $\pi_j$ until
(including) round $t$:
\begin{equation*}
  \Gamma_j(t) = \bigcup\limits_{i = t_j}^{t} \{\gamma(i)\}
\end{equation*}
As mentioned above, the only information that the player receives in
each round is whether its proposal was successful or not. While this
is little information, the game history until round $t$ can still help
the player to offer a better proposal for the next round. In
particular, at the end of round $t$ in phase $\pi_j$, the player knows
that $|\Gamma_j(t)\cap X(t+1)|=0$ because the adversary removes any
correct proposal from the target set. Therefore, the player should
not output any proposal which is in $\Gamma_j(t)$. On the other
hand, in round $t$ of phase $\pi_j$, any element in
$\Gamma_{j-1}\setminus \Gamma_j(t)$ is in $X(t+1)$ if and only if it
is in $Y_j$. For every round $t$ in phase $\pi_j$, we partition the
set $S$ into three disjoint sets: $\Gamma_j(t)$,
$\Gamma_{j-1}\setminus \Gamma_j(t)$, and
$S\setminus (\Gamma_{j-1}\cup \Gamma_j(t))$. Based on its knowledge,
the player can then deduce the following statements with respect to
the game history:
\begin{align*}
  \forall x\in \Gamma_j(t) : &\   x\not\in X(t+1), \\
  \forall x\in \Gamma_{j-1}\setminus \Gamma_j(t) : 
                             &\
                               x\in X(t+1) \Longleftrightarrow  x\in Y_j,\\
  \forall x\in S\setminus (\Gamma_{j-1}\cup \Gamma_j(t)) : 
                             &\
                               x\in X(t+1)\Longleftrightarrow x\in (Y_{j-1} \cup Y_j).
\end{align*}
Based on this information we can calculate the probability of success
for choosing a proposal from any of these three sets separately.

Fix some round $t$ in phase $\pi_j$ and let $h_j(t)$ be the number of
elements in $Y_j$ that are proposed correctly by the player until
round $t$ (in other words, $h_j(t)$ is the number of elements in $Y_j$
which are hit in some rounds in $[t_j,t]$). We will assume that the
adversary informs the player in each successful round whether the
proposal is from $Y_{j-1}$ or $Y_j$. As a player is free to ignore
this information, it only makes the adversary weaker (and the player
stronger). Any lower bound which holds with this additional
information therefore also holds without it. Knowing whether a
successful proposal in phase $\pi_j$ is from $Y_{j-1}$ or $Y_j$ lets
the player know $h_j(t)$ for all $t$ and $j$. Assume that
$x =\gamma(t+1)$ is the proposal chosen by the player for the round
$t+1$.

\paragraph{Case 1, \boldmath$x\in \Gamma_j(t)$}
In this case, the proposal of round $t+1$ is one of the earlier phase
$\pi_j$ proposals and it therefore cannot be in the target set. Thus,
we have
\[
\Pr\big(x\in X(t+1)|x\in \Gamma_j(t)\big) = 0.
\]
  
\paragraph{Case 2, \boldmath$x\in \Gamma_{j-1}\setminus \Gamma_j(t)$}
In this case the proposal is in the target set if and only if
$x\in Y_j$:
\[
\Pr\big(x\in X(t+1)\, |\, x\in \Gamma_{j-1}\setminus \Gamma_j(t)\big)
= 
\frac{|Y_j|-h_j(t)}{|S|-|\Gamma_j(t)|} \leq \frac{|Y_j|}{\beta
  -\varphi +1} \leq 
\frac{\ell}{\beta -\varphi +1} \leq \frac{2\ell}{\beta -
  2\varphi +1}.
\]

\paragraph{Case 3, \boldmath$x\in S\setminus (\Gamma_{j-1}\cup \Gamma_j(t))$}
Any element in $S\setminus (\Gamma_{j-1}\cup \Gamma_j(t))$ is in
$X(t+1)$ if and only if it is in $Y_j$ or in $Y_{j-1}$. Therefore, we
can estimate the probability of success for choosing any element in
$S\setminus (\Gamma_{j-1}\cup \Gamma_j(t))$ as the proposal for round
$t+1$ as follows:
\begin{eqnarray*}
  \Pr\big(x\in X(t+1)\, |\, x\in S\setminus \left(\Gamma_{j-1}\cup
  \Gamma_j(t)\right)\big) 
  & \leq & \frac{|Y_{j-1}\cup Y_j|-|\Gamma_{j-1}\cup \Gamma_j(t)|}{|S|-|\Gamma_{j-1}|-|\Gamma_{j}(t)|}\\ 
  & \leq & \frac{|Y_j|+|Y_{j-1}|}{\beta -2\varphi +1} \\
  & \leq & \frac{2\ell}{\beta -2\varphi +1}.
\end{eqnarray*}
From these three cases we can conclude that even when using all the
available information about the history, the probability of success
for any element from the selection set is at most
$\frac{2\ell}{\beta -2\varphi +1}$. Let $\bar{\mathcal{H}}$ be the
game history and let $\bar{x}$ be an arbitrary element from $S$. We
then have
\begin{equation}\label{equ:hitting-prob}	
  \Pr\left(\bar{x}\in X(t+1)\,|\, \bar{\mathcal{H}}\right) \leq \frac{2\ell}{\beta - 2\varphi +1}.
\end{equation}
Starting at any time $t$, let the random variable $R$ count the number
of rounds until the first successful round after round $t$. We have
\[
\E[R_1] \geq \frac{\beta -2\varphi +1}{2\ell}
\stackrel{(\varphi\leq\beta/3)}{\geq}
\frac{\frac{1}{3}\beta + 1}{2\ell}.
\]
Consequently, the expected time until having $\beta$ successful rounds
becomes $\Omega(\beta^2/\ell)$.
\end{proof}

\begin{lemma}\label{le:reduction}
  If algorithm $\mathcal{A}$ solves the global broadcast problem in
  any $T$-interval $1$-connected dynamic $n$-nodes network against a
  $1$-oblivious adversary in $f(n)=n^{O(1)}$ rounds in expectation for
  a sufficiently large value of $T$, then we can construct a player
  $\mathcal{P}$ to win the
  $(\lfloor n/2\rfloor-\ell, \ell , cT\ln n)$-periodic hitting game
  in expected $O(f(n)\log n)$ rounds, for some positive constants $c$
  and $\ell$.
\end{lemma}
\begin{proof}
  We construct a player $\mathcal{P}$ to simulate the execution of
  $\mathcal{A}$ on a particular $T$-interval $1$-connected dynamic
  $n$-node network (the target network). Then the player uses the transmitting behavior of
  the nodes in the simulation to generate guesses for playing the
  game. We start by defining the target network for a given instance
  of the $(\lfloor n/2\rfloor-\ell,\ell, cT\ln n)$-periodic hitting
  game.

\begin{figure}[t]
\centering
\begin{subfigure}{.4\textwidth}
  \centering
  \includegraphics[width=.9\textwidth]{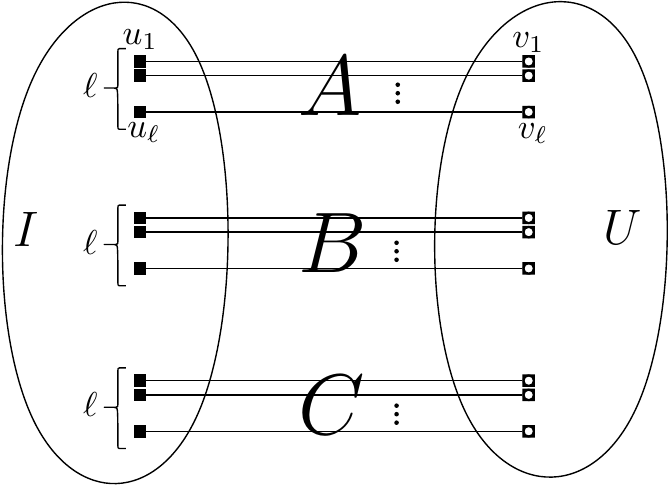}
  \caption{The core structure of the dynamic lower bound
    network. The edges labeled $A$, $B$, and $C$ are added in
    different phases. \label{fig:a}}
  \label{fig:sub1}
\end{subfigure}%
\hspace{1cm}
\begin{subfigure}{.4\textwidth}
  \centering
  \vspace{2cm}
  \includegraphics[width=.9\textwidth]{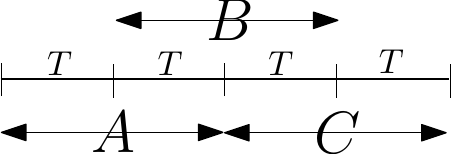}
  \caption{At least one edge exists between $I$ and $U$ for
    any $T$ consecutive rounds. \label{fig:b}}
  \label{fig:sub2}
\end{subfigure}
\caption{A snapshot of the dynamic network used in the hitting
      game simulation.}
\label{fig:lowerbound}
\end{figure}

  \paragraph{The Target Network}
  For the following discussion, we set
  $\beta:=\lfloor n/2\rfloor -\ell$ and $\varphi:=cT\log n$ to
  denote the size of the selection set and the length of a phase of
  the hitting game. We assume that we are given an instance of the
  $(\beta,\ell,\varphi)$-periodic hitting game. Based on how the
  hitting game develops, we define an $n$-node dynamic target
  network. We first describe the core (backbone) part of the
  network. The nodes of the dynamic network are defined as
  $V:=\set{0,\dots,n-1}$. We assume that node $0$ is the source and we
  identify the next $\beta$ nodes (i.e., the set $[\beta]$) with the
  selection set $S$ of the hitting game. Throughout the execution,
  node $0$ is connected to all nodes in $[\beta]$ and it is not
  connected to any other node. Throughout the simulation of the
  broadcast algorithm, we use $I$ and $U$ to denote the set of
  informed and uninformed nodes, respectively (a node is informed iff
  it knows the broadcast message $\mathcal{M}$). Clearly, as soon as
  the source node broadcasts $\mathcal{M}$, the set of informed nodes
  is $I=\set{0,\dots,\beta}$ and we thus have
  $U=\set{\beta+1,\dots,n-1}$. To simplify notation, we assume that
  already at the start of the simulation, all nodes in
  $\set{0,\dots,\beta}$ know $\mathcal{M}$ and thus, we start round
  $1$ with $I=\set{0,\dots,\beta}$ and $U =
  \set{\beta+1,\dots,n-1}$.
  We will assume that the number of uninformed nodes is always at
  least $2\ell$. As soon as it drops below, we stop carrying out the
  simulation.

  Throughout the simulation, we always assume that all nodes in $I$
  form a clique and all nodes in $U$ form a clique. Apart from this,
  the topology of the core network is determined by the target set of
  the hitting game that we are trying to win by simulating
  $\mathcal{A}$. Assume that in some round $r$ of the hitting game,
  the target set is $X(r)\subset[\beta]$. During the simulation, in
  the backbone network we then use the nodes in $X(r)$ as bridge nodes
  to connect the informed nodes to the uninformed nodes. Each node
  $x\in X(r)$ is connected to exactly one node $n(x) \in U$ such that
  each node in $U$ is connected to at most one node in $X(r)$. We
  assume that whenever a new node is added to $X(r)$, its neighbor in
  $U$ is chosen uniformly at random among all nodes in $U$ which are
  not already connected to a bridge node in $X(r)$. Note that the
  size of $X(r)$ is always at most $2\ell$ and because we assumed
  that $|U|\geq 2\ell$, we can always do such an assignment of bridge
  nodes. Whenever the player makes a successful guess $x\in X(r)$, we
  move $x$ to the set of informed nodes $I$ and we connect $x$ with
  all nodes in $I$ and disconnect it with all nodes in the remaining
  set $U$ of uninformed nodes. Note that in the hitting game, after a
  successful guess $x\in X(r)$, $x$ is also removed from the target
  set. The target network at any time is either the described core network 
  (backbone) or the complete graph $K_n$, for which the choice will be 
  explained later. 
 
  \paragraph{The Simulation}
  The simulation of the broadcast algorithm $\mathcal{A}$ is done in a
  round-by-round manner. As the dynamic topology used
  in the simulation depends on the target set of the hitting game, the
  player $\mathcal{P}$ of the hitting game does not know the dynamic
  topology. We need to show that $\mathcal{P}$ can still correctly
  simulate the behavior of the broadcast algorithm.

  As discussed above, we assume that at the beginning of the
  simulation, the set of informed nodes is $I = \set{0,\dots,\beta}$.
  Each round of $\mathcal{A}$ is simulated by $\mathcal{P}$ by making
  at most $c\ln n$ guesses in the hitting game. More specifically, a
  given round $r$ of $\mathcal{A}$ is simulated as follows.

  First, note that because we assume that the adversary is
  $1$-oblivious, $\mathcal{P}$ can base the graph of round $r$ on the
  states of all nodes at the beginning of the round. Hence, in
  particular, the graph of round $r$ can depend on the probability
  $p_v(r)$ with which each node $v\in I$ transmits in the given
  round. We define a round $r$ of $\mathcal{A}$ to be \emph{busy} if
  $\sum_{v\in I} p_v(r) > \frac{c}{2}\ln n$, otherwise a round $r$ is
  called \emph{free}. In a busy round, the network graph is assumed to
  be the complete graph $K_n$ and in a free round, the network graph
  is assumed to be exactly the backbone graph as described above. We
  assume that $\mathcal{P}$ always knows the set of informed nodes and
  because only informed nodes are allowed to transmit, $\mathcal{P}$
  can determine all messages which are transmitted in a round by
  simulating the random decisions of the nodes in $I$. We say that the
  simulated execution of $\mathcal{A}$ is \emph{bad} if either there
  is a free round in which more than $c\ln n$ nodes in $I$ decide to
  transmit a message or if there is a busy round in which
  exactly one node in $I$ decides to transmit. Otherwise, the
  simulated execution is called \emph{good}. If an execution
  turns out to be bad, $\mathcal{P}$ stops the simulation of
  $\mathcal{A}$ and simply continues making random guesses until it
  wins the hitting game. Note that the expected time to win the
  hitting game in this way is at most $O(\beta^2/\ell)$ as unless
  there have been at least $\Omega(\ell)$ successes during the last
  $\phi=O(\beta)$ rounds of the hitting game, the probability for a
  successful guess is always $\Omega(\ell/\beta)$. As long as the
  number of simulated rounds $f(n)$ of $\mathcal{A}$ is polynomial in
  $n$ (and thus also in $\beta$), for an arbitrary given constant
  $d>0$ and sufficiently large constant $c$, the probability to
  obtain a good execution is at least $1-1/n^d$. For sufficiently
  large constant $c$, the expected time to win the hitting game is
  therefore dominated by the expected time to win the game conditioned
  on the event that the simulation of $\mathcal{A}$ creates a good
  execution. In the following, we therefore assume that the generated
  broadcast execution is good. In the following, we also assume that
  in the current phase (of length $\varphi$) of the hitting game,
  there are still at least $c\ln n$ guesses that can be made. If this
  is not the case, player $\mathcal{P}$ first makes a sequence of
  unsuccessful guesses to finish the phase ($\mathcal{P}$ can for
  example repeat the last guess it has made before to make sure it is
  not successful). As we assumed that $T$ is sufficiently large, we
  can assume that $\varphi \gg c\ln n$ and therefore we only waste a
  small fraction of all guesses by doing this.

  Let us first assume that a simulated round $r$ of $\mathcal{A}$ is
  busy. As in this case, the communication network is a complete graph
  and since we assume that in a good execution, no node gets isolated 
  (transmits alone),
  every node receives silence and we therefore do not need to simulate
  any receive behavior. In this case, we also do not make any guesses
  in the hitting game. If round $r$ is free, the number of nodes that
  transmit is between $0$ and $c\ln n$. First recall that the nodes
  in $I$ are fully connected and $\mathcal{P}$ can therefore clearly
  simulate their receive behavior. Further, let
  $Z(r)\subseteq [\beta]$ be the set of nodes in $\beta$ which are
  transmitting in round $r$. For each $z\in Z(r)$, player
  $\mathcal{P}$ uses $z$ as a guess in the hitting game. Note that
  because there are at most $c\ln n$ guesses to be made and because
  we assumed that there are still at least $c\ln n$ guesses in the
  current phase of the hitting game, during making the guesses for all
  $z\in Z(r)$, we do not change the phase (and thus the target set) in
  the hitting game. The node $z$ therefore is a bridge node connecting
  $I$ to a node $n(z)\in U$ in round $r$ of the broadcast algorithm if
  and only if $z$ is a successful guess. In that case, $n(z)$ is a
  uniformly random node in $U$. Hence, if $z$ is a successful guess,
  $\mathcal{P}$ chooses $n(z)$ uniformly at random in $U$ and it moves
  $n(z)$ from $U$ to $I$. Note that $z$ is also removed from the
  target set, and all connections of $n(z)$ to nodes in the remaining set $U$
  are removed. Note also that by choosing $n(z)$ uniformly at random
  in $U$, player $\mathcal{P}$ does not only simulate the randomness
  of the broadcast algorithm, but it also simulates the randomness of
  the adversary. As long as the execution is good, in the given
  dynamic network, the broadcast algorithm informs a new node if and
  only if one of the bridge nodes $v$ transmits in a free round. For
  any bridge node $v$ which transmits in a free round, the
  corresponding uninformed bridge node $n(v)$ gets informed. The
  described simulation therefore correctly simulates the broadcast
  algorithm and it informs a new node if and only if it makes a
  correct guess. As we only stop the simulation once the number of
  uninformed nodes drops below $2\ell$, we need at least
  $n-(\beta+1)-2\ell+1 = \lceil n/2\rceil - \ell \geq \beta$
  successful guesses and thus win the game to stop the simulation.

  It remains to show that the dynamic network used in the simulation
  is $T$-interval connected. Every $\varphi$ guesses and thus after at
  least $\varphi / (c\ln n)$ rounds of the simulation, we add $\ell$
  new edges connecting some $v\in I$ and $n(v)\in U$. As long as $v$
  is not used as a guess, such an edge remains for $2\varphi$
  guesses. As long as it is always guaranteed that one of these edges
  survives the next $2\varphi$ guesses (and thus at least
  $2\varphi / (c\ln n)$ rounds), the network is at least
  $T=\varphi/(c\ln n)$-interval connected. Hence, the network is not
  guaranteed to be $T$-interval connected if there is a phase $j$ in
  the hitting game such that all the elements of the set $Y_j$ added
  to the target set at the beginning of phase $j$ are successfully
  guessed by the end of phase $j+1$. Recall that in this case, the
  player also wins the game and therefore the claim of the lemma
  follows.
\end{proof}

For $k=1$, the statement of our main lower bound theorem (Theorem
\ref{thm:lowerbound}) now directly follows by combining Lemmas
\ref{le:gamelb} and \ref{le:reduction}.

\begin{lemma}\label{lemma:simplelower}
  For every constant $\eps>0$ and every $T\leq n^{1-\eps}$, the
  expected time to solve single-message broadcast in $T$-interval
  $1$-connected radio networks against a $1$-oblivious adversary is at
  least $\Omega\big(\frac{n^2}{\log n}\big)$.
\end{lemma}
\begin{proof}
  For the sake of contradiction let us assume that $\mathcal{A}$ can
  solve broadcast for any $T$-interval 1-connected network in
  $f(n)=o(n^2/\ell \log n)$ rounds, then based on
  Lemma~\ref{le:reduction} a player can solve any instance of the
  $(n/2, \ell , cT\log n)$-periodic hitting game in $o(n^2/\ell)$
  rounds which contradicts Lemma~\ref{le:gamelb} and this proves the
  necessity of $\Omega (n^2/\ell \log n)$ rounds to solve broadcast in
  any $T$-interval $1$-connected network. Based on Lemma~\ref{le:gamelb}, 
  by choosing sufficiently large constant value for $\ell$ the claimed lower bound follows.
\end{proof}

In order to obtain Theorem \ref{thm:lowerbound}, we need to generalize
the above result from $T$-interval $1$-connected networks to
$T$-interval $k$-connected networks for arbitrary $k\geq 1$. As shown
below, this can be achieved by using a simple generic reduction.

\begin{theoremR}{thm:lowerbound}
  For every constant $\eps>0$ and every $T\leq (n/k)^{1-\eps}$, the
  expected time to solve single-message broadcast in $T$-interval
  $k$-connected radio networks against a $1$-oblivious adversary is at
  least $\Omega\big(\frac{n^2}{k^2\log n}\big)$.
\end{theoremR}
\begin{proof}
  Given an $n$-node graph $G$, let $H_k(G)$ be the graph which is
  obtained by replacing each node of $G$ by a clique of size $k$ and
  by replacing each edge $\set{u,v}$ of $G$ by a complete bipartite
  subgraph $K_{k,k}$ between the two $k$-cliques representing $u$ and
  $v$. If $G$ is connected, in order to disconnect $H_k(G)$ by
  deleting some nodes, we need to completely remove at least one of
  the $k$-cliques representing the nodes of $G$. Hence, if $G$ is
  connected, $H_k(G)$ is $k$-vertex connected. It follows in the same
  way that if we have a dynamic graph $G_1,\dots,G_t$ which is
  $T$-interval $1$-connected, the dynamic graph
  $H_k(G_1),\dots,H_k(G_t)$ is $T$-interval $k$-connected. Even if all
  nodes of such a graph $H_k(G_i)$ know to which of the cliques
  representing the nodes of $G_i$ they belong, solving broadcast in
  the dynamic graph $H_k(G_1),\dots,H_k(G_t)$ cannot be easier than
  solving broadcast in $G_1,\dots,G_t$. If each graph $H_k(G_i)$ has
  $N$ nodes, the graphs $G_i$ have $N/k$ nodes and the
  claimed lower bound directly follows by applying Lemma
  \ref{lemma:simplelower}.
\end{proof}


%% file: impossibility.tex
\section{Impossibility of Broadcast Against a \boldmath$0$-Oblivious
  Adversary}\label{sec:impossibility}

In this section, we prove the impossibility result that we stated in
Section \ref{sec:contributions}. We show that unless $T$ is almost
equal to $n$, the global broadcast problem cannot be solved in the
presence of a $0$-oblivious adversary, even for very large vertex
connectivity $k$.

\begin{theoremR}{thm:impossibility}
  For any $k\geq 1$ and any $T<n-k$, it is not possible to solve
  single-message broadcast in $T$-interval $k$-connected radio
  networks against a $0$-oblivious adversary.
\end{theoremR}
\begin{proof}
  We show that a strongly adaptive adversary (i.e., a $0$-oblivious
  adversary) can apply a simple strategy to prevent any algorithm from
  solving the global broadcast problem in a $T$-interval $k$-connected
  network, where $T<n-k$.  Consider the following adversary strategy
  to determine the sequence of network topologies.
		
  The adversary partitions the $n$ nodes into two distinct sets $A$
  and $B$, such that $A$ includes the source node and is of size
  $T+k$, and $B$ is of size $n-(T+k)$.  Since $T<n-k$, there exists at
  least one node in $B$. Note that at the beginning, no node in $B$
  knows the broadcast message $\mathcal{M}$ (or anything about
  $\mathcal{M}$). Because the adversary is $0$-oblivious, in each
  round $r$, it can determine the graph after all nodes have made
  their random decisions. It can therefore determine the graph based
  on which nodes transmit.

  If in a round $r$, either $0$ nodes transmit or at least $2$ nodes
  transmit, the network graph is chosen to be the complete graph. Note
  that in such a round, there is either silence or all nodes
  experience a collision. In both cases, all listening nodes receive
  $\bot$ and therefore no node in $B$ can learn something about
  $\mathcal{M}$.

  If in a round, exactly one node $v$ in $A$ transmits, the network
  graph consists of all edges except the edges connecting $v$ to nodes
  in $B$. Like this, also in this case all nodes in $B$ receive $\bot$
  and they therefore cannot learn something about $\mathcal{M}$.

  It remains to show that the given dynamic graph is $T$-interval
  $k$-connected. During the whole execution, $A$ is a clique
  consisting of $T+k$ nodes. Hence, $A$ is a $T$-interval
  $k$-connected network. To show that whole $n$-node network is also a
  $T$-interval $k$-connected network, it is sufficient to show that
  for any node $v\in B$, in any $T$ consecutive rounds, there exist at
  least $k$ fixed edges from $v$ to the nodes in $A$. To do so, fix
  some arbitrary time interval of $T$ consecutive rounds. During the
  time interval, there are at most $T$ rounds in which exactly one
  node transmits. Therefore, because $|A|=T+k$, there are at least $k$
  nodes in $A$ which do not transmit alone during the given time
  interval. The edges from these $k$ nodes to all nodes in $B$ are
  therefore available throughout the $T$ rounds. Therefore throughout
  any interval of $T$ rounds, each node in $B$ is connected to a set
  of at least $k$ nodes in $A$. Consequently, the constructed dynamic
  network is $T$-interval $k$-connected.
\end{proof}

Notice that at least for store-and-forward algorithms even collision
detection does not help to overcome the impossibility result. The
$0$-oblivious adaptive adversary knows the random choices of the
algorithm in the current round and it can thus prevent any
progress.

We also note that the above result turns out to be tight in the
following sense. If $T \geq n-k$, global broadcast can be solved. If
in each round, every node independently tries to broadcast with some
probability (say $1/n$), if $T + k \geq n$ there is a non-zero
probability (it may be very small) that $T$ different nodes are
isolated in $T$ consecutive rounds. Consider an interval of $T$ rounds
and let $I$ and $U$ be the sets of informed and uninformed nodes at
the beginning of this interval. From $T$-interval $k$-connectivity, we
get that there are at least $k$ nodes in $I$ which are stably
connected to nodes in $I$ throughout the $T$ rounds. Before broadcast
is solved, we have $|I|\leq n-1$ and if in the $T$ rounds, $T$
different nodes in $I$ are isolated, at least one of the $k$ nodes
stably connected to $U$ gets isolated and we can therefore make
progress. Note that for $T=n-k$, the probability for making progress
might be exponentially small, resulting in an exponential running time
for the broadcast problem. Note however also that once
$T\geq cn\log n$ for a sufficiently large constant $c$, it is not hard
to show that broadcast can be solved in polynomial time against a
$0$-oblivious adversary and if $T$ is larger than $cn\log^2n$ for a
sufficiently large constant $c$, it is shown in \cite{dualgraph}, that
it can be solved in time $O\left(n\log^2 n\right)$.


\hide{Anisur: 

*** \todo{ to be included in the end of Subsection 4.1} ***

We can easily generalize the above result into the case of $k$ connected networks. In particular, we show an $\Omega \left(n^2/k^2\tau\log n\right)$ lower bound for the global broadcast in $T$-interval $k$-connected radio networks against a $\tau$-oblivious adversary for $\tau \geq 1$.      

\begin{theorem}\label{thm:lowerbound-k-connect}
There exists a $T$ interval $k$ connected radio network such that any randomized algorithm requires at least $\Omega \left(\frac{n^2}{k^2\tau\log n}\right)$ time to solve global broadcast against a $\tau$-oblivious adversary, where $1 \leq T \leq \left(\frac{n}{k}\right)^{1-\eps}$ for some constant $\eps$. 
\end{theorem}
\begin{proof}[Proof sketch]
Given a $n$-node graph $G$, we can always construct a new graph of $nk$ nodes replacing each node of $G$ by a clique of size $k$. Assume the above $\tau$-dual $T$-interval $1$-connected radio network of $n/k$ nodes. Then replace every node in this graph by a clique of size $k$ and for each edge, add complete bipartite connection between two cliques. See the Figure~\ref{fig:fragment-graph}. Clearly this new graph has $n$ nodes in total and it is $k$ connected. Hence the $\tau$-dual $T$-interval $1$-connected graph of size $n$ is equivalent (i.e., has similar structure) to a $\tau$-dual $T$-interval $k$-connected graph of size $n/k$. Therefore, we can apply the same reduction from the $(\alpha,\beta,l,\varphi)$-periodic hitting game to get a lower bound  $\Omega \left(\frac{n^2}{k^2\tau\log n}\right)$ of the global broadcast problem in $\tau$-dual $T$-interval $k$-connected radio networks.  
\end{proof}

The impossibility result is also valid even with the collision detection. The strongly adaptive adversary knows the random choices of the algorithm in the current round and  
}
